\documentclass[a4paper, 12pt]{article}

\usepackage[round]{natbib}

\usepackage[utf8]{inputenc}

\usepackage{float}
\usepackage{graphicx}
\usepackage{subcaption}

\usepackage[linktocpage,colorlinks=true,citecolor=black,linkcolor=black,urlcolor=black]{hyperref}
\usepackage[a4paper, total={6in, 8in}]{geometry}

\usepackage{amsmath}
\usepackage{amsthm}
\usepackage{amsfonts}
\usepackage{amssymb}
\usepackage{bm}
\usepackage{mathdots}
\usepackage{mathtools}

\usepackage{enumerate}

\usepackage{appendix}
\usepackage[all,cmtip]{xy}
\usepackage{multirow}

\newcommand\abs[1]{\left|#1\right|}

\newtheorem{lemma}{Lemma}
\newtheorem{theorem}{Theorem}

\newcommand{\parcoef}[2]{\phi_{#1}^{(#2)}}
\newcommand{\parcor}[1]{\beta_{#1}}

\newcommand{\zzero}{z_{0}}
\newcommand{\zzeroinv}{\zzero^{-1}}
\newcommand{\Pn}[1][n]{P_{#1}}
\newcommand{\Qn}[1][n-m]{Q_{#1}}

\newcommand{\Phin}[1][n]{\Phi_{#1}}

\newcommand{\pmbeta}{\gamma} % "plus/minus beta"

\newcommand{\zbar}{\bar{z}}
\newcommand{\rr}{r}

\title{
  Partial autocorrelation parameterisation of models with unit
  roots on the unit circle
}

\author{Jamie Halliday \\ University of Manchester
        \and
        Georgi N. Boshnakov \\ University of Manchester
}

\begin{document}

\maketitle

\section{Introduction}
\label{S:saruma-intro}

%\input{sarima-intro}

% Nonstationary ARMA models are frequently used for modelling time series.

Let $\{Y_{t}\}$ be a time series whose evolution can be described by the equation
\begin{equation}
  \label{eq:ARUMA}
  U(B) \phi(B) Y_t = \theta(B) \varepsilon_t
  .
\end{equation}
\citet{TiaoTsay1983} refer to this model as a nonstationary ARMA model.  \citet{HuangAnh1990}
call this model autoregressive unit root moving average (ARUMA), see also
\citet{WoodwardGrayElliott2017}.  Here $\{\varepsilon_{t}\}$ is white noise, $B$ is the
backward shift operator and all roots of the polynomials $\phi(z)$ and $\theta(z)$ are
outside the unit circle. The nonstationary part is specified by the polynomial
$U(z) = 1 - U_1 z - U_2 z^2 - \dots - U_d z^d$ whose all roots have moduli~1 (i.e., lie on
the unit circle). Traditionally the polynomial $U(z)$ does not have coefficients to be
estimated.  This is the case, for example, for the familiar ARIMA and seasonal ARIMA (SARIMA)
models obtained when $U(z) = (1-B)^{d}$ and $U(z) = (1-B)^{d}(1-B^s)^{d_{s}}$, respectively.

\citet{TiaoTsay1983} and \citet{TsayTiao1984} study (iterative) ordinary least squares
procedures for estimation of such models and, in particular, show how the unit roots can be
estimated consistently. 

For time series data it is typical to consider whether seasonal trends appear. This behaviour
is easily captured by the existing models by allowing further polynomials to appear in the
model with the appropriate power transformation of $B$ to account for the
seasonality. Standard example is the \textit{SARIMA} class of models, mentioned above. The
operator $(1-B^{s})^{d_{s}}$ however is sometimes too crude and may be inpractical when the
number of seasons, $s$, is large or in the case of multiple seasons.  A more flexible class
of models is obtained by replacing it with a operator containg only some harmonics of $1/s$.
With a seasonal extension, we refer to this class of models as \textit{SARUMA}. Here is a
symbolic representation:
% Consider a univariate time series $Y_t$ following the SARUMA model given below.
\begin{equation}
  \label{eq:SARUMA}
  U_s(B^s) U(B) \phi_s(B^s) \phi(B) Y_t = \theta_s(B^s) \theta(B) \varepsilon_t,
\end{equation}
where $U_s(z)$ is a seasonal polynomial of degree $d_s$ where all roots are unit, $\phi_s(z)$
is a seasonal autoregressive polynomial of degree $p_s$, $\theta_s(z)$ is a seasonal moving
average polynomial of degree $q_s$ such that all roots of $\phi_s(z)$ and $\theta_s(z)$ lie
outside the unit circle. The remaining terms are as in Equation~\ref{eq:ARUMA}.
%%%% \begin{itemize}
%%%% \item $\{ \varepsilon_t \}$ is a white noise error term,
%%%% \item $B$ is the backshift operator,
%%%% \item $U(z)$ is a polynomial of degree $d$ where all roots are unit, 
%%%% \item $U_s(z)$ is a seasonal polynomial of degree $d_s$ where all roots are unit,
%%%% \item $\phi(z)$ is an autoregressive polynomial of degree $p$ where all roots lie outside the
%%%%   unit circle,
%%%% \item $\phi_s(z)$ is a seasonal autoregressive polynomial of degree $p_s$ where all roots lie
%%%%   outside the unit circle,
%%%% \item $\theta(z)$ is a moving average polynomial of degree $q$ where all roots lie outside
%%%%   the unit circle to ensure the model is invertible,
%%%% \item $\theta_s(z)$ is a seasonal moving average polynomial of degree $q_s$ where all roots
%%%%   lie outside the unit circle to ensure the model is invertible.
%%%% \end{itemize}
We also require that there are no common roots between the $\phi_s(z^s)\phi(z)$ and
$\theta_s(z^s)\theta(z)$ components of the model. In practice, it is sometimes useful to
factor $U(z)$ and $U_{s}(z)$ into further factors in order to obtain more meaningful and/or
manageable models.

In principle the SARUMA model can be written in the form of model~\eqref{eq:ARUMA} by
expanding $U_{s}(B^{s})U(B)$ and estimate it using the OLS method of \citet{TsayTiao1984} but
this looses any parsimony that might be achievable otherwise. 

Parameterisations of stationary models through partial autocorrelations are widely used in
the stationary case but for unit root models partial autocorrelations are not
defined. Nevertheless, we show that partial autocorrelations equal to $\pm1$ naturally
describe multiplicative ARUMA models and neatly fit with the standard practice of fitting
ARIMA models. We continue to call them partial autocorrelations though they do not have the
usual statistical interpretation and are purely a parameterisation of the polynomial on the
left-hand side of Equation~\eqref{eq:SARUMA}.

The transformation from partial autocorrelations to polynomial coefficients is unique, so
residuals and sums of squares are easily available and estimation is posible.

In this paper we obtain the algebraic properties of the partial autocorrelations in the
context of unit roots. The main result is that if a partial autocorrelation sequence contains
some values equal to $1$ or $-1$, then it can be split at these values into sequences
each of which represents the partial autocorrelations of a factor of the overall polynomial
on the left-hand side of the model. 
A separate paper will discuss the details of the estimation procedure
and its properties. An implementation is provided by \citet[][function
\texttt{sarima}]{Rsarima}.

%%%% \section{From ts.tex: Box-Jenkins}
%%%% \label{S:ts-BoxJenkins}
%%%% 
%%%% \input{ts_BJ}
%%%% 
%%%% %%  \subsection{From ts.tex: State space}
%%%% %%  \label{S:ts-stateSpace}
%%%% %% 
%%%% %%  \input{ts_Kalman}
%%%% %% 
%%%% %%  \newpage{}

%----------------------------------------------------------------------------------------
%	  INTRODUCTION
%----------------------------------------------------------------------------------------

\section{Levinson-Durbin algorithm and its inverse}

The use of partial autocorrelations as a parameterisation of autoregressive (AR) stationary
models and stable filters is well established.  For stationary AR models there is a
one-to-one map between the autoregressive parameters and the partial autocorrelations. The
partial autocorrelations have a clear statistical meaning in this case. The one-to-one map
allows to think of the partial autocorrelations also as an alternative way to parameterise
the coefficients of the associated autoregressive polynomial.

For a stationary process $\{X_{t}\}$, let $\parcoef{1}{n},\dots,\parcoef{n}{n}$ be the
partial prediction coefficients for the best linear predictor,
$\parcoef{1}{n}Y_{t} + \dots + \parcoef{n}{n}Y_{t-n+1}$, of $Y_{t+1}$ based on the latest
available $n$ observations.  Let $\parcor{1}, \parcor{2},\ldots$ be the partial
autocorrelations. It is convenient to define $\parcor{0} = 1$.  Consider also the polynomial
\begin{align*}
  1 - \parcoef{1}{n}z - \dots - \parcoef{n}{n}z^{n}
  .
\end{align*}
The statistical meaning of the partial autocorrelations and partial prediction coefficients
is not really needed for the exposition below but gives context.

The Levinson-Durbin recursions \citep{BrockwellDavis1991} can be used to compute the partial
prediction coefficients from the partial autocorrelations, as follows:
\begin{align}
\label{eq:pacf2ar1}
     \parcoef{n}{n} &= \parcor{n} 
  \\ 
\label{eq:pacf2ar2}  
  \parcoef{k}{n} &= \parcoef{k}{n - 1} - \parcor{n} \parcoef{n - k}{n - 1}
                      \qquad{}
                      \text{$k = 1, \dots, n - 1$}
  \\  & \text{(for $n = 1, 2, \ldots $)}  \nonumber
  .                     
\end{align}
It is evident that the transformation from partial autocorrelations to partial coefficients is uniquely
defined without the need to put restrictions on
$\parcor{1}, \dots, \parcor{n}$. Note that, strictly speaking, the Levinson-Durbin algorithm
contains an additional step at each $n$ for computing the partial autocorrelation from
autocorrelations, which we don't need since we start with partial autocorrelations. 

The recursions can be arranged in reverse order to compute the partial autocorrelations 
from the partial coefficients $\parcoef{1}{m}, \dots \parcoef{m}{m}$:
\begin{align}
     \label{eq:ar2pacf1}
     \parcor{n}  &= \parcoef{n}{n}
  \\
     \label{eq:ar2pacf2}
     \parcoef{k}{n - 1}     &= (\parcoef{k}{n} + \parcor{n} \parcoef{n - k}{n} )
                              / (1 - \parcor{n}^{2}) 
                      \qquad{}
                      \text{$k = 1, \dots, n - 1$}
  \\  & \text{(for $n = m, m - 1, \ldots, 1$)} \nonumber
  .                     
\end{align}
At the end we have $\parcor{1},\dots,\parcor{n}$. Detailed discussion of several variants of
the Levinson-Durbin algorithm is given by \citet{porat2008digital}.

Of course, the inverse recursion will work only if $|\parcor{k}| \neq 1$ for
$k = 1, \dots, n$.  In that case the relationship between the two sets of coefficients is
one-to-one. The case $\abs{\parcor{k}} > 1$ is not of interest to us here. Our aim is to show
that allowing some of the partial autocorrelations to be equal to one provides a very natural
parameterisation for models with arbitrary unit roots, including seasonal ARIMA models.
Since partial autocorrelations uniquely determine the filter coefficients, this means that
residuals can be computed and so a non-linear least squares estimation of the unit root
filter can be performed.

Some further insight can be obtained by noticing that the equations are paired for $k$ and
$n-k$:
\begin{align*}
   \parcoef{k}{n} &= \parcoef{k}{n - 1} - \parcor{n} \parcoef{n - k}{n - 1}  \\
   \parcoef{n-k}{n} &= \parcoef{n-k}{n - 1} - \parcor{n} \parcoef{k}{n - 1}
                      \qquad{}
                      \text{$k = 1, \dots, [n/2]$}
 .                     
\end{align*}
If $n$ is even and $k = n/2 = n-k$ the two equations can be reduced to
\begin{equation*}
  \parcoef{n/2}{n}
    = \parcoef{n/2}{n - 1} - \parcor{n} \parcoef{n/2}{n - 1}
    = \parcoef{n/2}{n - 1}(1 - \parcor{n})
    .
\end{equation*}
In particular, if $\parcor{n} = 1$ then $\parcoef{n/2}{n} = 0$ and if $\parcor{n} = -1$ 
then $\parcoef{n/2}{n} = 2\parcoef{n/2}{n - 1}$. It is also obvious that when $k \neq n/2$ 
that $\parcoef{k}{n} = -\parcoef{n-k}{n}$ when $\parcor{n} = 1$ and 
$\parcoef{k}{n} = \parcoef{n-k}{n}$ when $\parcor{n} = -1$. For example, when $n=2$, 
the above gives $\parcoef{1}{2} = 0$ if $\parcor{2} = 1$ and the polynomial must be 
$1 - z^{2}$. When $\parcor{2} = -1$ then $\parcoef{1}{2} = 2\parcoef{1}{1} = 2\parcor{1}$ 
and the polynomial is $1 - 2\parcor{1}z + z^{2}$, which generates a pair of complex roots.

In what follows we show how polynomials can be separated after the occurrence of
a partial autocorrelation value of unit magnitude and show that sequence of partial 
autocorrelations ending with a unit value produces a polynomial that contains only
roots on the unit circle. This methodology can be used to define each polynomial in 
Equation~\eqref{eq:SARUMA}.

%-------------------------------------------------------------------------------
%	 PARAMETERISATION
%-------------------------------------------------------------------------------

\section{Parameterisation using partial autocorrelations}

Let $\parcor{k}$, $k = 1, 2, \ldots$, be a sequence of partial autocorrellations.
Define polynomials $\Pn(z)$ by
\begin{equation} \label{eq:Pn}
  \Pn(z) = \sum_{k = 1}^n \parcoef{k}{n} z^k
     ,\quad \text{for $n = 1, 2, \ldots$,}
     \qquad \Pn[0](z) = 0
  ,
\end{equation}
where $\parcoef{k}{n}$ are the partial coefficients obtained from
$\parcor{1},\dots,\parcor{n}$, using
Equations~\eqref{eq:pacf2ar1}--\eqref{eq:pacf2ar2}.
Our main interest is in the positions of the zeroes of  the polynomials
\begin{equation*}
  \Phin(z) = 1 - \Pn(z)
           = 1 - \sum_{k = 1}^{n} \parcoef{k}{n} z^k
     ,\quad \text{for $n = 0, 1, 2, \ldots$}
  .
\end{equation*}

It is well known that if the coefficients of the polynomial $\Phin{(z)}$ are obtained from
partial autocorrelations $\parcor{1}, \dots, \parcor{n}$, such that $\abs{\parcor{i}} < 1$
for $i=1,\dots,n$, then all zeroes of the polynomial $\Phin{(z)}$ are outside the unit circle
(i.e., have moduli greater than~1). In particular, their product has modulus larger than~1.

What happens if $\abs{\parcor{i}} < 1$ for $i=1,\dots,n-1$, but $\parcor{n} = \pm1$? We
formulate the result as a lemma. It is hardly new but not easily available.
\begin{lemma} \label{le:unitRoots}
  If $\abs{\parcor{i}} < 1$ for $i=1,\dots,n-1$, $\parcor{n} = \pm1$, then all zeroes,
  $z_1, \dots, z_{n}$, of the polynomial $\Phin{(z)} = 1 - P_{n}(z)$ are on the unit circle
  (i.e., $\abs{z_i} = 1$ for $i = 1, \dots, n$).
\end{lemma}
One way to show this is to notice that in that case the Vietta formulas imply that the
product of the zeroes of $\Phin(z)$ is $\pm{1}$.  Then let $\parcor{n}^{(i)} \to \parcor{n}$,
$\abs{\parcor{n}^{(i)}} < 1$ for $i=1,2,\ldots$ and consider the sequence of polynomials
$\Phi_{n}^{(i)}(z)$, $i=1,2,\ldots$.  Since the zeroes of polynomials are continuous
functions of their coefficients, and hence the partial autocorrelations, the zeroes of
$\Phi_{n}^{(i)}(z)$ converge to the zeroes of $\Phin(z)$. But all zeroes of
$\Phi_{n}^{(i)}(z)$ are strictly outside the unit circle, so their limits (the zeroes of
$\Phin(z)$) are outside or on the unit circle. This means that their product can be equal
to~1 only if all of them have modulus~1.

The following relation between the polynomials $\Pn(z)$ can be
obtained from the Levinson-Durbin recursions.  Let $n \ge 2$. For
general $z$, multiply Equation~\eqref{eq:pacf2ar2} by $z^k$ for
$k = 1, \dots, n-1$, and sum to obtain
\begin{equation*}
  \sum_{k=1}^{n-1} \parcoef{k}{n} z^k
  = 
  \sum_{k=1}^{n-1} \phi_k^{(n-1)} z^k  - \parcor{n} \sum_{k=1}^{n-1} \phi_{n-k}^{(n-1)} z^k
  .
\end{equation*}
Using the definition of the polynomial $\Pn(z)$ and $\parcor{n} = \parcoef{n}{n}$,
this can be written as
\begin{equation*}
  \Pn(z) - \parcor{n} z^n = P_{n-1}(z) - \parcor{n} z^n P_{n-1}(z^{-1})
  ,
\end{equation*}
which after rearranging becomes
\begin{equation}
\label{eq:PD1n}
(1 - \Pn(z)) = (1 - P_{n-1}(z)) - \parcor{n} z^n \left( 1 - P_{n-1}(z^{-1}) \right)
.  
\end{equation}
The above equation was derived for $n \ge 2$ but it holds also, trivially, for $n = 1$.
Note that the coefficients of the polynomial  $z^n \left( 1 - P_{n-1}(z^{-1})\right)$ are
those of $(1 - \Pn(z))$ in reverse order. 

In general, the polynomials $1 - P_{n}(z)$, $n = 1, 2,\ldots$, do not have common zeroes.  A
remarkable exception, particularly important for unit root models, is given by the following
lemma. It shows that if $\zzero$ is such that it and $\zzeroinv$ are both zeroes of the
polynomial $1 - P_{m}(z)$, then they are also zeroes of the polynomials $1 - P_{n}(z)$ for
all $n\ge m$.
\begin{lemma}
  Let $\zzero$ be such that $1 - \Pn[m](\zzero) = 0$ and $1 - \Pn[m](\zzeroinv) = 0$ for
  some $m \in \mathbb{Z}^+$.  Then  $1 - \Pn(\zzero) = 0$ and
  $1 - \Pn(\zzeroinv) = 0$ for any $n \geq m$.
\end{lemma}
\begin{proof}

Setting $n = m + 1$ in Equation~\eqref{eq:PD1n} gives
\begin{equation}
%\label{eq:PD1}
  (1 - \Pn[m+1](z))
  = (1 - \Pn[m](z)) - \parcor{m+1} z^{m+1} \left( 1 - \Pn[m](z^{-1}) \right).  
\end{equation}
If $z = \zzero$ or $\zzeroinv$, then both terms on the right-hand side of the last equation
are zero, by the assumptions of the lemma.  Hence, the left-hand side is also zero, i.e.
$1 - P_{m+1}(z_0) = 0$ and $1 - P_{m+1}(z_0^{-1}) = 0$. So, the claim of the lemma holds for
$n = m + 1$.  But Equation~\eqref{eq:PD1n} holds also for $n > m + 1$, so the proof can be
completed by induction.
\end{proof}

The following corollary concerning roots on the unit circle is of primary interest for our
purposes. Indeed, complex roots of polynomials with real coefficients come in complex
conjugate pairs. Moreover, if $\abs{z_{0}} = 1$ then $z_{0}^{-1} = \zbar_{0}$. So, in this
case $1 - P_{m}(z_0) = 0$ implies  $1 - P_{m}(z_0^{-1}) = 0$ and we have:
% It holds from the complex conjugate root theorem \citep{BusamFreitag2009}, 
% since the polynomials considered in this paper have real coefficients.

\begin{lemma} \label{le:z}
  If $\abs{z_0} = 1$ and $1 - P_{m}(z_0) = 0$
  then $1 - \Pn(\zzero) = 0$ and $1 - \Pn(\zzeroinv) = 0$ for any $n \geq m$.
\end{lemma}

A useful consequence of Lemma~\ref{le:z} is the following result.

\begin{lemma} \label{le:zall} If all roots, $z_1, \dots, z_{m}$, of the polynomial
  $1 - P_{m}(z)$ are on the unit circle (i.e., $\abs{z_i} = 1$ for $i = 1, \dots, m$), then
  $1 - P_{m}(z)$ is a factor of $1 - \Pn(z)$ for any $n \geq m$.
\end{lemma}
\begin{proof}
  Since the roots have moduli equal to 1 and $1 - P_{m-1}(z)$ has real
  coefficients, it follows from Lemma~\ref{le:z} that
  $z_{1},\dots,z_{m}$ are roots of $1 - \Pn(z)$ for all $n\ge m$,
  hence the result.
\end{proof}

Lemma~\ref{le:zall} shows that if $\Pn(z)$ is the polynomial generated from the
partial autocorrelation sequence
$\parcor{1},\dots,\parcor{m},\parcor{m+1},\dots,\parcor{n}$, where
$\parcor{m}=\pm 1$ and $|\parcor{m+i}|<1$ for $i=1,\dots,n-m$, then $1 - \Pn(z)
= (1 - \Pn[m](z)(1 - T(z))$, where $T(z)$ is some polynomial. It turns out that
$\parcor{m+1},\dots,\parcor{n}$ are, up to possible sign changes, the partial
autocorrelations generating the polynomial $T(z)$.  Our main result in this
section states the complete result.

\begin{theorem}[Main result] \label{thm:polyDecomp}
  %% 2019-09-13 changing the formulation
  %%
  %% Let $\Pn(z)$ be the polynomial generated from the partial autocorrelations
  %% $\parcor{1},\dots,\parcor{m},\parcor{m+1},\dots,\parcor{n}$, where
  %% $|\parcor{i}|\le 1$, for $i = 1,\dots, m - 1$, $\parcor{m}=\pm 1$, and
  %% $|\parcor{m+i}| \le 1$ for $i = 1, \dots, n - m$.  Let $\pmbeta_{i} =
  %% (-1)^{d_{+}}\parcor{m+i}$, $i = 1, \dots, n - m$, where $d_{+}$ is the number
  %% of roots of $1 - P_{m}(z)$ equal to $+1$.

  Let 
  $\parcor{1},\dots,\parcor{m},\parcor{m+1},\dots$, be partial autocorrelations, such that 
  $|\parcor{i}|\le 1$, for $i = 1,\dots, m - 1$, $\parcor{m}=\pm 1$, and
  $|\parcor{m+i}| \le 1$ for $i \ge 1$.
  Let $\Pn(z)$ be the polynomials defined by Equation~\eqref{eq:Pn}. 
  Let also $\pmbeta_{i} = (-1)^{d_{+}}\parcor{m+i}$, $i \ge 1$, where $d_{+}$ is
  the number of zeroes of $1 - P_{m}(z)$ equal to $+1$.

  Then, for each $n \ge m + 1$, $(1 - \Pn(z)) = (1 - \Pn[m](z))(1 - \Qn(z))$,
  where the polynomial $\Qn(z)$ is generated from the partial autocorrelations
  $\pmbeta_{1},\dots,\pmbeta_{n - m}$.
\end{theorem}
\begin{proof}
Changing $n$ to $l$ in Equation~\eqref{eq:PD1n} and summing from $m+1$ to $n$ we
obtain
\begin{align*}
  \sum_{l=m+1}^n (1-P_l(z))
  &= \sum_{l=m+1}^{n} (1-P_{l-1}(z)) - \sum_{l=m+1}^n \parcor{l} z^l \left( 1 - P_{l-1}(z^{-1})\right).
  \\ &= \sum_{l=m}^{n-1} (1-P_{l}(z)) - \sum_{l=m+1}^n \parcor{l} z^l \left( 1 - P_{l-1}(z^{-1})\right).
\end{align*}
After cancelling the common terms in the two sides of the equation and
rearranging, we get
\begin{align}
\label{eq:PD2}
  (1 - \Pn(z))
  &= (1 - P_{m}(z)) - \sum_{l=m+1}^n \parcor{l} z^l \left( 1 - P_{l-1}(z^{-1})\right).
    \nonumber
\end{align}
In particular, for $m = 0$ we have
\begin{align}
  (1 - \Pn(z))
  &= (1 - P_{0}(z)) - \sum_{l=1}^n \parcor{l} z^l \left( 1 - P_{l-1}(z^{-1})\right).
    \nonumber
  \\ &= (1  - \sum_{l=1}^n \parcor{l} z^l \left( 1 - P_{l-1}(z^{-1})\right).
\end{align}
  By Lemma~\ref{le:unitRoots} all roots of the polynomial $1 - P_{m}(z)$ are on the unit 
  circle. Let $d_{+}$ and $d_{-}$ be the number of roots equal to $+1$ and $-1$, 
  respectively. The remaining $2r$ roots are complex conjugate pairs, 
  $\alpha_{i}, \alpha_{i}^{-1}$, $i=1,\dots,r$, where $\alpha_{i}^{-1}$ is the complex 
  conjugate of $\alpha_{i}$ since $|\alpha_{i}|=1$.  Obviously, $m = d_{+} + d_{-} + 2r$. 
  We have
  \begin{align*}
    1 - \Pn[m](z)
    &= \left( 1 - z \right)^{d_+} \left( 1 + z \right)^{d_-} 
           \prod_{i=1}^r \left(1 - \frac{z}{\alpha_i} \right) \left(1 - \alpha_i z
      \right)
    \\ &= \left( 1 - z \right)^{d_+} \left( 1 + z \right)^{d_-} 
           \prod_{i=1}^r \left(1 - (\frac{1}{\alpha_i} + \alpha_i) z + z^{2} \right)
    .
  \end{align*}
From this we get
\begin{align}
  1 - \Pn[m](z^{-1})
  &= \left( 1 - z^{-1} \right)^{d_+} \left( 1 + z^{-1} \right)^{d_-} 
   \prod_{i=1}^r \left(1 - \frac{z^{-1}}{\alpha_i} \right) \left(1 - \alpha_i z^{-1} 
    \right), \nonumber \\
   &= z^{-d_{+}}(z - 1)^{d_{+}}
     z^{-d_{-}}(z + 1)^{d_{-}}
     \prod_{i=1}^r z^{-2} \left(z - \frac{1}{\alpha_i} \right) \left(z - \alpha_i \right)
      \nonumber \\
   &= z^{-d_{+}}(z - 1)^{d_{+}}
     z^{-d_{-}}(z + 1)^{d_{-}}
     z^{-2r} \prod_{i=1}^r  \left(z^{2} - (\frac{1}{\alpha_i} + \alpha_{i} ) z + 1 \right)
      \nonumber \\
   &= z^{-m} (-1)^{d_+} \left( 1 - z \right)^{d_+} \left( 1 + z \right)^{d_-} 
     \prod_{i=1}^r \left(1 - \frac{z}{\alpha_i} \right) \left(1 - \alpha_i z \right)
     \nonumber \\
   \label{eq:PD4}
   &= z^{-m} (-1)^{d_+} \left( 1 - P_{m}(z) \right)
\end{align}
Together with Equation~\eqref{eq:PD2} (with $n = m + 1$) this gives
\begin{align*}
1 - \Pn[m+1](z) &= 1 - \Pn[m](z) -
                  \parcor{m+1} z^{m+1} \left( z^{-m} (-1)^{d_+} \left( 1 - P_{m}(z) \right) \right),
  \\ &= \left( 1 - P_{m}(z) \right) \left( 1 - (-1)^{d_+} \parcor{m+1} z \right).
  \\ &= \left( 1 - P_{m}(z) \right) \left( 1 - \pmbeta_{1} z \right)
  .
\end{align*}
Therefore, when $n = m + 1$, $1 - P_{m}(z)$ is a factor of $1 - P_{m+1}(z)$ and,
moreover, we have the explicit factorisation with $Q_{1}(z) = \pmbeta_{1} z$.

For the general case, let $n > m + 1$ and assume that the claim is true for all $l < n$. 
Concentrate on the case $l = n$ and let $1 - Q_{l-m}(z)$ represent the polynomial remaining 
after division of $1 - P_{l}(z)$ by $1 - P_{m}(z)$ for $l>m$, so that 
\begin{equation}
\label{eq:PD5}
1 - P_{n}(z) = \left( 1 - P_{m}(z) \right) \left( 1 - Q_{n-m}(z) \right)
\end{equation}
Starting from Equation~\eqref{eq:PD1n}, and with the help of 
Equation~\eqref{eq:PD4},
\begin{align}
  1 - P_{n}(z) 
  &= 1 - P_{m}(z) - \sum_{l=m+1}^n \parcor{l} z^l \left( 1 - P_{l-1}(z^{-1}) \right) 
    \nonumber \\
  &= 1 - P_{m}(z) - \sum_{l=m+1}^n \parcor{l} z^l  \left( 1 - P_{m}(z^{-1}) \right) 
    \left( 1 - Q_{l-1-m}(z^{-1}) \right), \nonumber \\
  &= 1 - P_{m}(z) - \sum_{l=m+1}^n \parcor{l} z^{l-m} (-1)^{d_+}   \left( 1 - P_{m}(z) \right) 
    \left( 1 - Q_{l-1-m}(z^{-1}) \right), \nonumber \\
  &= (1 - P_{m}(z))\left(1 - \sum_{l=m+1}^n  (-1)^{d_+} \parcor{l}  z^{l-m} 
    \left( 1 - Q_{l-1-m}(z^{-1}) \right) \right), \nonumber \\
  &= (1 - P_{m}(z))\left(1 - \sum_{l=m+1}^n   \pmbeta_{l-m}  z^{l-m} 
    \left( 1 - Q_{l-1-m}(z^{-1}) \right) \right), \nonumber \\
  &= (1 - P_{m}(z))\left(1 - \sum_{k=1}^{n-m} \pmbeta_{k}  z^{k} 
    \left( 1 - Q_{k-1}(z^{-1}) \right) \right),
  \label{eq:PD6} 
  %%%% &=
  %%%%   \left( 1 - P_{m1}(z) \right) \left(1 - (-1)^{d_+} \sum_{k=1}^{n-m} \parcor{k} z^{k} 
  %%%%   \left( 1 - Q_{k-1}(z^{-1}) \right) \right).
\end{align}

Equation~\eqref{eq:PD6} shows that $1 - P_{m}(z)$ is a factor of $1 - P_{n}(z)$ for some 
$n>m$ and moreover, by comparing it with Equation~\eqref{eq:PD5} we can see that
\begin{equation*}
1 - Q_{n-m}(z) = 1 - \sum_{l=1}^{n-m} \gamma_l z^l \left(1 - Q_{l-1}(z^{-1}) \right),
\end{equation*}
where $\gamma_l = (-1)^{d_+} \parcor{l}$. Notice the similarities between this equation 
and Equation~\eqref{eq:PD2}. $1 - Q_{n-m}(z)$ is of the same form as the original 
polynomial $1-\Pn(z)$ except that the original partial autocorrelation coefficients 
$\parcor{k}$ have been replaced by $\gamma_k$.

By induction, 
% $1 - P_{m}(z)$ is a factor of $1 - P_{n}(z)$ for all $n \geq m$.
the claim of the theorem is proved.
\end{proof}

If there are more partial autocorrelations with modulus~1,
Theorem~\ref{thm:polyDecomp} can be applied recursively to get a factorisation
of the unit root polynomials.

\begin{theorem}
  Let $m_{1} < m_{2} < \dots < m_{r}$, be positive integers such that
  $\abs{\parcor{m_{i}}} = 1$, $i=1,\dots,\rr$.
  Then for each $n \ge m_{r} + 1$
  \begin{align*}
    (1 - \Pn(z))
    &= (1 - \Pn[m_{1}](z))(1 - \Pn[m_{2} - m_{1}](z))
      \cdots
      (1 - \Pn[m_{r} - m_{r-1}](z))
      (1 - \Qn[n-m_{r}](z))
      ,
  \end{align*}
  where the polynomials $(1 - \Pn[m_{i}](z))$ are obtained from the partial
  autocorrelations $\parcor{i}$, $i=m_{i-1}+1,\dots,m_{i}$ with adjusted signs
  as given by Theorem~\ref{thm:polyDecomp} (applied recursively)
  and the polynomial $\Qn(z)$ is generated from the partial autocorrelations
  $\pmbeta_{1},\dots,\pmbeta_{n - m_{r}}$.
\end{theorem}

There are a number of ways to use Theorem~\ref{thm:polyDecomp} in
modelling. The most transparent and useful is given by the following result.

%% TODO: @jamie @georgi - think of something better than 'ARUMA result'
\begin{theorem}[ARUMA result] \label{thm:polyDecomp2}
  Let $n > m$ and $\parcor{1},\dots,\parcor{m},\parcor{m+1},\dots,\parcor{n}$,
  be partial autocorrelations, such that $|\parcor{i}|\le 1$, for $i = 1,\dots,
  m - 1$, $\parcor{m}=\pm 1$, and $|\parcor{m+i}| < 1$ for $i = 1, \dots, n -
  m$.  Let $\Pn(z)$ be the polynomials defined by Equation~\eqref{eq:Pn}.  

  Then $(1 - \Pn(z)) = (1 - \Pn[m](z))(1 - \Qn(z))$, where all zeroes of $(1 -
  \Pn[m](z))$ are on the unit circle and all zeroes of $(1 - \Qn(z))$ are
  outside the unit circle. Further, $1 - \Pn[m](z)$ is generated by
  $\parcor{1},\dots,\parcor{m}$ and $1 - \Qn(z))$ by $\pmbeta_{1},\dots,\pmbeta_{n
    - m}$, where $\pmbeta_{i} = (-1)^{d_{+}}\parcor{m+i}$, $i = 1,\dots, n - m$
  and $d_{+}$ is the number of zeroes of $1 - P_{m}(z)$ equal to $+1$.
\end{theorem}
\begin{proof}
  The factorisation $(1 - \Pn(z)) = (1 - \Pn[m](z))(1 - \Qn(z))$ follows from
  Theorem~\ref{thm:polyDecomp}. By Lemma~\ref{le:unitRoots} all zeroes of
  $(1 - \Pn[m](z))$ are on the unit circle. Further, $(1 - \Qn(z))$ since by
  Theorem~\ref{thm:polyDecomp} they are generated by partial autocorrelations
  $|\pmbeta_{i}| < 1$, $i = 1,\dots, n - m$, which have the same moduli as
  $\parcor{m+1},\dots, \parcor{n}$.
\end{proof}

Theorems~\ref{thm:polyDecomp} and~\ref{thm:polyDecomp2} fit nicely with the
standard practice of applying unit root and/or seasonal unit root filters
(represented here by the polynomial $1 - \Pn[m](z)$) to make a time series
stationary and then fitting a stationary model to the filtered time series.  The
unit root filters are typically chosen in advance.  Our results allow for
estimating the unit root filter. In the simplest case, $\parcor{m}$ (where $m$
is as in Theorem~\ref{thm:polyDecomp2}) is fixed to $\pm1$ and the remaining
partial autocorrelations are estimated using non-linear optimisation in the unit
cube.

% This sentence was in the version for saruma.tex:
%
% This is discussed in more detail in the following sections.

Recall that for the SARUMA model $\Phi(z) = 1 - P_n(z)$. From the results above, we know 
that $\Phi(z)$ decomposes into $\left( 1 - P_m(z) \right) \left( 1 - Q_{n-m}(z) \right)$ 
if all roots of $\left( 1 - P_m(z) \right)$ are on the unit circle. We express 
$\left( 1 - P_m(z) \right)$ as $U(z)$, the unit root polynomial. If no unit partial 
autocorrelation values remain in $\left( 1 - Q_{n-m}(z) \right)$ then this corresponds to 
the stationary $\phi(z)$. Otherwise, the unit root polynomials can be iteratively 
separated and stored as a product in $U(z)$. When $U(z)$ contains all nonstationary 
aspects of the model, the Levinson-Durbin recursion can be used to generate the 
coefficients of $U(z)$ by fixing the final coefficient to $\pm 1$. For example, say that 
$U(z)$ is of degree $d$. The remaining partial autocorrelations can be used to estimate 
the coefficients in $\phi(z)$, starting from $\beta_{d+1}$ and after multiplication with 
$(-1)^{d_+}$.

Firstly assume, without loss of generality, that all seasonal polynomials 
can be dropped ($U_s(z) = \phi_s(z) = \theta_s (z) \equiv 1$). Then the resulting ARUMA 
model can be written
\begin{equation}
\label{eq:nsARMA}
\Phi(B) Y_t  = \theta(B) \varepsilon_t.
\end{equation}
Furthermore, define the polynomial $P_n(z)$ as
\begin{equation*}
  P_n(z) = \sum_{k = 1}^n \phi_k^{(n)} z^k
  ,\quad \text{for $n = 1, 2, \ldots$,}
  \qquad \Pn[0](z) = 0
  ,
\end{equation*}
so that $\Phi(z) = 1 - P_n(z)$ with  $n = p + d$.

We will discuss the details and the properties of an estimation procedure for
ARUMA models based on the results here in a separate paper.  An implementation
can be found in package `sarima' \citep{Rsarima}.

%----------------------------------------------------------------------------------------
%	  BIBLIOGRAPHY
%----------------------------------------------------------------------------------------

\bibliographystyle{abbrvnat}
\bibliography{sarima}

\begin{thebibliography}{8}
\providecommand{\natexlab}[1]{#1}
\providecommand{\url}[1]{\texttt{#1}}
\expandafter\ifx\csname urlstyle\endcsname\relax
  \providecommand{\doi}[1]{doi: #1}\else
  \providecommand{\doi}{doi: \begingroup \urlstyle{rm}\Url}\fi

\bibitem[Bistritz(1996)]{Bistritz1996}
Y.~Bistritz.
\newblock {Reflections on Schur--Cohn matrices and Jury--Marden tables and
  classification of related unit circle zero location criteria}.
\newblock \emph{Circuits, Systems, and Signal Processing}, 15\penalty0
  (1):\penalty0 111--136, 1996.
\newblock \doi{10.1007/BF01187696}.

\bibitem[Boshnakov and Halliday(2022)]{Rsarima}
G.~N. Boshnakov and J.~Halliday.
\newblock \emph{sarima: Simulation and Prediction with Seasonal ARIMA Models},
  2022.
\newblock URL \url{https://CRAN.R-project.org/package=sarima}.
\newblock https://geobosh.github.io/sarima/ (doc).

\bibitem[Brockwell and Davis(1991)]{BrockwellDavis1991}
P.~J. Brockwell and R.~A. Davis.
\newblock \emph{Time Series: Theory and Methods}.
\newblock Springer Series in Statistics. Springer New York, second edition,
  1991.

\bibitem[Huang and Anh(1990)]{HuangAnh1990}
D.~Huang and V.~V. Anh.
\newblock Estimation of the non-stationary factor in {ARUMA} models.
\newblock \emph{Journal of Time Series Analysis}, 14\penalty0 (1):\penalty0
  27--46, 1990.
\newblock \doi{10.1111/j.1467-9892.1993.tb00128.x}.

\bibitem[Porat(1994)]{porat2008digital}
B.~Porat.
\newblock \emph{Digital processing of random signals: theory and methods}.
\newblock Prentice-Hall information and system sciences series. PTR
  Prentice-Hall, 1994.

\bibitem[Tiao and Tsay(1983)]{TiaoTsay1983}
G.~C. Tiao and R.~S. Tsay.
\newblock Consistency properties of least squares estimates of autoregressive
  parameters in {ARMA} models.
\newblock \emph{The Annals of Statistics}, 11\penalty0 (3):\penalty0 856--871,
  1983.
\newblock \doi{10.1214/aos/1176346252}.

\bibitem[Tsay and Tiao(1984)]{TsayTiao1984}
R.~S. Tsay and G.~C. Tiao.
\newblock Consistent estimates of autoregressive parameters and extended sample
  autocorrelation function for stationary and nonstationary {ARMA} models.
\newblock \emph{Journal of the American Statistical Association}, 79\penalty0
  (385):\penalty0 84--96, 1984.
\newblock \doi{10.1080/01621459.1984.10477068}.

\bibitem[Woodward et~al.(2017)Woodward, Gray, and
  Elliott]{WoodwardGrayElliott2017}
W.~Woodward, H.~Gray, and A.~Elliott.
\newblock \emph{Applied Time Series Analysis with R (2nd ed.)}.
\newblock CRC Press, 2017.
\newblock URL \url{https://doi.org/10.1201/9781315161143}.

\end{thebibliography}

%----------------------------------------------------------------------------------------
%	  APPENDIX
%----------------------------------------------------------------------------------------

\appendix

\section{Stable polynomials}
\label{S:saruma-stablePoly}

In signal processing, the partial autocorrelations (multiplied by $-1$) are known as
\textit{reflection coefficients} (RCs) and play an important role in determining the zero
locations of a polynomial with complex coefficients. Let $p(z)$ denote such a polynomial of
degree $n$, then
\begin{equation*}
  p(z) = \sum_{i=0}^n p_i z^i.
\end{equation*}
The polynomial is called \textit{stable} if all roots of the polynomial lie outside the unit
circle. The RCs contain the necessary information regarding the locations of roots with
respect to the unit circle and the following theorem holds \citet{Bistritz1996}:
\begin{theorem}
\label{thm:ZeroLocation}
A polynomial $p(z)$ with a well-defined set of RCs $\{\beta_k \}_{k=1}^n$, 
$\abs{\beta_k} \neq 1$, has $\nu$ roots inside the unit circle and $n-\nu$ roots outside 
the unit circle, where $\nu$ can be calculated by counting the number of negative terms 
in the sequence
\begin{equation*}
\nu = n_- \{ q_n, q_{n-1}, \dots q_1 \}
\end{equation*}
whose members are defined by
\begin{equation*}
q_k = \prod_{i=n}^k (1 - \beta_i^2), \qquad k = n, \dots, 1.
\end{equation*}
\end{theorem}
An immediate consequence of Theorem~\ref{thm:ZeroLocation} is that necessary and 
sufficient conditions for stability (or $\nu = 0$) are
\begin{equation*}
\abs{\parcor{k}} < 1, \qquad k = 1, \dots, n.
\end{equation*}
The result is formulated for RCs but holds also for partial autocorrelations since it
involves only their moduli and squares.

\end{document}